\documentclass[10pt]{IEEEtran}
\usepackage{amsthm,amsmath,amssymb,color,graphicx,overpic,bbm, subcaption}
\newcommand{\subparagraph}{}\usepackage[compact]{titlesec}
\usepackage[font=small,labelfont=bf]{caption}
\usepackage{cite}
\theoremstyle{plain}

\theoremstyle{definition}
\theoremstyle{plain}

\theoremstyle{definition}

\providecommand{\definitionname}{Definition}

\providecommand{\lemmaname}{Lemma}
\providecommand{\theoremname}{Theorem}
\providecommand{\remarkname}{Remark}
\newtheorem{definition}{Definition}
\newtheorem{lemma}{Lemma}
\newtheorem{theorem}{Theorem}

\newtheorem{assumption}{Assumption}
\newtheorem{corollary}{Corollary}
\def\QED{\mbox{\rule[0pt]{1.3ex}{1.3ex}}}

\newcommand{\prob}[1]{\mathbb{P}\!\left(#1\right)}

\begin{document}
\title{Selective Fair Scheduling over Fading Channels}

\newcounter{one}
\setcounter{one}{1}
\newcounter{two}
\setcounter{two}{2}

\author{\IEEEauthorblockN{Apostolos Destounis$^{\fnsymbol{one}}$,  Georgios S. Paschos$^{\fnsymbol{one}}$, David Gesbert$^{\fnsymbol{two}}$
		\\}
	\IEEEauthorblockA{$^{\fnsymbol{one}}$ Mathematical and Algorithmic Sciences Lab, France Research Center, Huawei Technologies Co. Ltd. \\
		$^{\fnsymbol{two}}$ EURECOM\\
		emails: $^{\fnsymbol{one}}$ \{apostolos.destounis, georgios.paschos\}@huawei.com, $^{\fnsymbol{two}}$ gesbert@eurecom.com
	}
}
\maketitle


\begin{abstract}
Imposing fairness in resource allocation incurs a loss of system throughput, known as the Price of Fairness ($PoF$). In wireless scheduling, $PoF$ increases when serving users with very poor channel quality because the scheduler wastes resources trying to be fair. This paper proposes a novel resource allocation framework to rigorously address this issue. We introduce \emph{selective fairness}: being fair only to selected users, and improving $PoF$  by momentarily blocking the rest. 
We study the associated admission control problem of finding the user selection that minimizes $PoF$ subject to selective fairness, and show that this combinatorial problem can be solved efficiently if the feasibility set satisfies a condition; in our model it suffices that the wireless channels are stochastically dominated. 
Exploiting selective fairness,  we design a stochastic framework where we minimize $PoF$ subject to an SLA, which ensures that an \emph{ergodic subscriber} is served frequently enough. In this context, we propose an online policy that combines the \emph{drift-plus-penalty} technique with \emph{Gradient-Based Scheduling} experts, and we prove it achieves the optimal $PoF$. Simulations show that our intelligent blocking outperforms by 40$\%$ in throughput previous approaches which satisfy the SLA by blocking low-SNR users.
\end{abstract}


\vspace{-0.1in}
\section{Introduction}

Throughput efficiency and fairness is a well-explored fundamental tradeoff in wireless communications \cite{Kushner}.
The status quo is to use \emph{opportunistic schedulers} to exploit the fading peaks 
and strike a balance between throughput and fairness 
 \cite{Kushner,tseDumbAntennas}. While  fairness is important, it   is also known to negatively impact the system throughput, a phenomenon quantified by the {Price of Fairness} ($PoF$) \cite{pof}. 
In wireless downlink systems $PoF$ increases steeply when
 a base station attempts to serve  users with {\em unreasonably} poor channel quality, for instance users trying to access spectrum from strongly shadowed areas (e.g. building basements, tunnels), remote areas (cell-edge users), or operating with ill-functioning or sub-standard RF equipments. 
Our simulations show a 10$\%$-40$\%$ throughput degradation when adding 1-10  users with SNR 20dB less. 
Since  service quality is anyway low in these cases, the current approach is to block users when their SNR drops below a threshold.
\emph{We propose a framework for proactively  blocking users, based on optimizing the $PoF$ subject to fairness and a probabilistic service guarantee. For the same quality level, our scheme yields a $40\%$ total throughput gain over the current approach, unraveling significant room for optimization which was not previously explored.}

We start with a $K$-user scheduling problem over a wireless fading channel. The service must be fair, but our key idea is that it is allowed temporarily to exclude some users from service. In this context, we introduce a novel fairness objective called \emph{selective fairness}: a subset of users $S\subseteq \{1,\dots, K\}$ is fairly treated, while the remaining users $S^c$ receive no service. Imposing selective fairness for $s_{\min}$ users as a constraint, we consider the minimization of $PoF$ (equivalent to system throughput maximization). Solving this problem is non-trivial because the throughput contribution of a user to a fairness-constrained system is very complicated. Mathematically, the problem is of combinatorial nature, potentially involving  the solution  of $O(2^K)$  large convex programs. We show, however, that if the system satisfies the \emph{subspace monotonicity} property, then the problem can be solved efficiently. We further prove that the subspace monotonicity property is satisfied when the fading channels are \emph{stochastically ordered}, a practical case of interest.
Then  we propose an online policy, referred to as \emph{selective GBS}, which is based on $O(K)$ number of experts, i.e.,  online simulated policies that provide insight for good scheduling decisions. The throughput vector obtained by selective GBS is shown to converge to  the optimal solution of the $PoF$ minimization.

The initial theoretical  framework assumes we know how many users to block, which is impractical. It is more reasonable to regulate the blocking according to a probabilistic service guarantee over multiple scheduling problems. The system will block users when their channel quality happens to be very poor, but also ensure that they are served most of the times they attempt to access the service. In the second part of the paper we extend selective fairness to a stochastic setting, where the blocking is controlled by a virtual queue evolving across scheduling problems.
 Combining the queue with selective GBS, we design an online policy, referred to as \emph{Online Selective Fair (OSF)} scheduler, that maximizes system performance while satisfying the probabilistic guarantee and being $\alpha$-fair to selected users. This provides a rigorous framework to alleviate the problem of $PoF$ in wireless scheduling.

\subsection{Related work}

The concept of opportunistic scheduling dates back to 1995 \cite{knopp}. The  \emph{Gradient-based Scheduler} (GBS) was proposed and analyzed early in the 2000s, cf.~\cite{Vijay, stolyar, Huang06}. It has been shown to provide a stochastic approximation of  the optimal solution of the Network Utility Maximization problem \cite{Kushner, Vijay, stolyar}, while it can also provide  good short-term fairness performance by using  discounting factors when averaging \cite{tseDumbAntennas, eryilmaz17}. For these reasons, and also for its great simplicity, GBS has become the de facto scheduling policy in 3G base stations \cite{nomadic}. To capture frequency-time resource blocks in LTE systems,  GBS was later extended to a multichannel version \cite{Huang06, leith,capozzi13}, keeping the original properties. Prior work has shown how to extend  GBS  to handle systems with many  antennas, which is necessary for  the 4G and emerging 5G wireless networks \cite{neelycaire}. As of today GBS remains the prominent practical scheduler, therefore we restrict our approach to be backward compatible with  GBS. 

It is anecdotally known  that most operators tune the GBS to achieve \emph{proportional fairness}, a tradeoff between maximizing total throughput and providing equal throughput shares to all users \cite{Kushner,kelly}. When some users have much lower channel quality, maintaining proportional fairness results in a  reduction of system throughput because the base station is forced to assign a great number of resources to them with small throughput return. 
Prior work  has analyzed the phenomenon of $PoF$ in general resource allocation problems \cite{pof,bertsimas}.  
In this paper, we are interested to judiciously exclude some users from service in order to minimize $PoF$. 
To our knowledge, there exist no prior work studying the optimization of $PoF$. 
We find that in systems with $5\%$ of users with very low channel quality, $PoF$ optimization can improve total throughput by $40\%$ over the existing state of art which simply blocks low-SNR users below a threshold. This gain is attributed to the fact that the optimal set of users to block at each realization varies from case to case, and it can not be determined by a simple predefined SNR threshold for blocking.

Our work is related to the literature of \emph{admission control in stochastic networks} and \emph{call admission control in cellular networks}, however there are important differences. Works in stochastic networks focus mainly on admitting  fractions of elastic traffic 
cf.~\cite{admissionLTE, Li05}, while in our case we admit \emph{a number} of users. 
Regarding the combinatorial call admission control for cellular networks, the most relevant work to ours is \cite{Bonald03}, where the authors discuss 
 admission control and blocking rates with opportunistic scheduling and evaluate the performance of simple mechanisms. In this paper, we derive a low-complexity policy that explicitly maximizes the cell spectral efficiency subject to fairness and a blocking constraint.   
More broadly, a differentiating factor of our work from existing literature is the consideration of special fairness constraints.
\section{System model}
\subsection{Wireless downlink}

We consider a \emph{wireless downlink}  with one transmitting base station and $K$ receiving users. Time is slotted $t=1,2,\dots$, and at each time slot, the 
base station can transmit data to \emph{one of the users} with the ultimate goal to optimize the time-average data transmissions. An extension to simultaneous service of multiple users is possible via \cite{Huang06, leith,capozzi13}, but we avoid it in the interest of presentation clarity.

If user $k$ is scheduled at slot $t$, data is sent to this user at a transmission rate  $R_k(t)\in \mathcal{R}$, where $\mathcal{R}=\{r_1,r_2,\dots\,r_L\}$ is a finite set of possible transmission rates. The vector $\boldsymbol R(t)\in \mathcal{R}^K$ is random, i.i.d. over time, and its  randomness is attributed to the wireless channel fading, thus  it is  independent of the past choices of the base station. The realization of $\boldsymbol R(t)$ is provided to the base station just before the scheduling decision is made, as it is customary in contemporary systems.

Let $I_k^{\pi}(t)\in\{0,1\}$ denote  the scheduling decision at time $t$ regarding  user $k$  under scheduling policy $\pi$, where $I_k^{\pi}(t)$ is 1 if user $k$ is scheduled, and 0 otherwise. A policy that schedules only one user at each slot, i.e., satisfies $\sum_k I_k^{\pi}(t)\leq 1,~ \forall t$, is called \emph{feasible}, and we denote with $\Pi$ the set of all feasible policies.
In our model, the base station always has available data for each user,\footnote{If we replace GBS with a \emph{max-weight}-type policy, it is possible to generalize our work to stochastic arrivals using the framework in \cite{georgiadis06}.} therefore the instantaneous  rate of data received by user $k$ during slot $t$  is
$\mu_k^{\pi}(t)=R_k(t)$ if $I_k^{\pi}(t)=1$ and zero otherwise, 
the   accumulated user  throughput at $t$ 
\[
\overline{x}_k^{\pi}(t)=\frac{\sum_{\tau=1}^t \mu_k^{\pi}(\tau)}t,
\]
and the \emph{user $k$  throughput}  is  
${x}_k=\liminf_{t\to\infty} \overline{x}_k^{\pi}(t)$. 
The vector of user throughputs, denoted with $\boldsymbol{x}$, is our key performance metric. 

\begin{definition}[Feasible throughputs]
The set of feasible throughputs  $\mathcal{X}$ is the set of all  throughput vectors $\boldsymbol{x}$ that can be achieved by any policy in $\Pi$.
\end{definition}


{
Let $p_{\boldsymbol r}=\prob{\boldsymbol R(t)=\boldsymbol r}, \boldsymbol r\in\mathcal{R}^K$ denote the probability distribution of channel rate vectors. 
By considering all slots with $\boldsymbol R(t)=\boldsymbol r$ and scheduling  user $k$ with probability $\phi_{k\boldsymbol{r}}$, taking liminf leads to the \emph{convex} set $\mathcal{X}$:
\begin{equation}\label{eq:feasible}
\mathcal{X}=\left\{\boldsymbol x\geq \boldsymbol 0 \left| 
\begin{array}{l}
x_k=\sum_{\boldsymbol r\in\mathcal{R}^K}\phi_{k\boldsymbol{r}}p_{\boldsymbol r} r_k \\
0\leq \phi_{k\boldsymbol{r}} \leq 1\\
 \sum_k\phi_{k\boldsymbol{r}}=1, \forall \boldsymbol r\in \mathcal{R}^K
 \end{array}\right. \right\}.
\end{equation}

\begin{theorem}[Feasible throughputs \cite{georgiadis06}]
The set of feasible throughputs  $\mathcal{X}$ is given by  set \eqref{eq:feasible}.
\end{theorem}

}

 \subsection{Efficiency and fairness}
 

Resource allocation  with multiple users  involves  two important goals: (i) to operate the system at high efficiency, and (ii) to allocate resources in a fair manner. Typically, these two goals are conflicting. Consider the  2-user example of  figure \ref{fig:setX}, where $\mathcal{X}$ is the shown gray area and the marginal user throughputs satisfy $X_1^{\max}>X_2^{\max}$. 
Point A corresponds to the maximum sum throughput--it is the point in $\mathcal{X}$ that maximizes $x_1+x_2$.
An arising issue with point A however, is that user 2 receives zero throughput, which is unfair. 
The fairest point is C, which ensures that the users receive the maximum possible equal throughputs. In this case however, the total system throughput is significantly reduced. 
In practice, engineers desire a tradeoff between the two extremes, A and C. Point B, known as \emph{proportional fairness} provides such a tradeoff. 
Next we formalize the fairness notions of interest.

 \begin{figure}[t!]
\includegraphics[width=6cm]{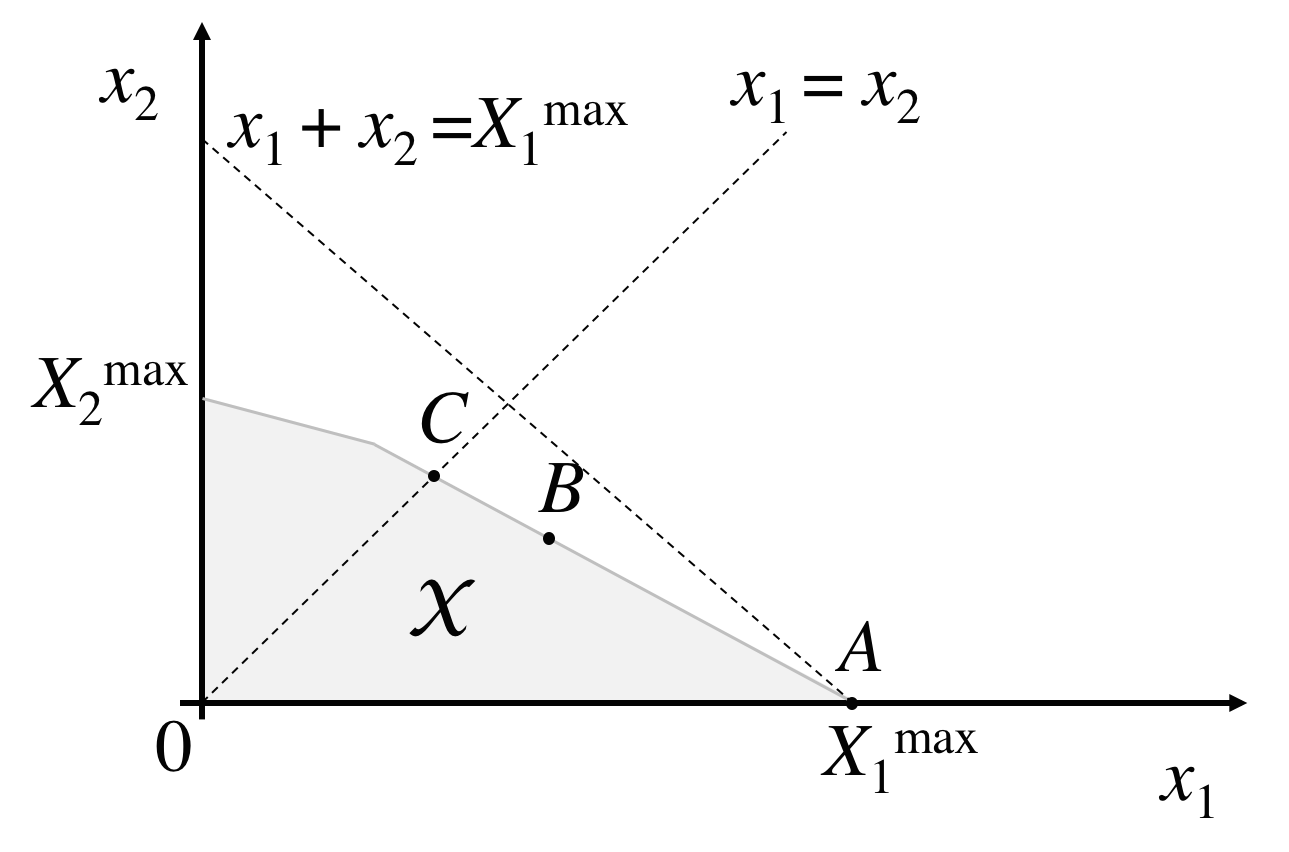}
\centering
\caption{Illustration of feasible set $\mathcal{X}$ for a 2-user downlink.
\vspace{-0in}}
\label{fig:setX}
\end{figure}



\begin{definition}[Fairness objectives]
Let $\mathcal{X}$ be a convex set of feasible throughputs. 
\begin{itemize}
\item A throughput  vector $\boldsymbol x \in \mathcal{X}$ is called \emph{max-sum-throughput} if for any  vector $\boldsymbol y \in \mathcal{X}$ it holds:
\[
\sum_{k=1}^K x_k \geq \sum_{k=1}^K y_k.
\]
\item A throughput vector $\boldsymbol x\in\mathcal{X}$ is called \emph{max-min} fair if for any vector $\boldsymbol y\in\mathcal{X}$ it holds:
\[
y_i>x_i~~\Rightarrow~~ \exists~ j: y_j<x_j\leq x_i.
\]
\item A throughput vector  $\boldsymbol x\in \mathcal{X}$ is called \emph{proportionally} fair if for any vector $\boldsymbol y\in\mathcal{X}$ it holds:
\[
\sum_k \frac{y_k-x_k}{x_k}\leq 0.
\]
\end{itemize}
\end{definition}

Since $\mathcal{X}$ in \eqref{eq:feasible} is a closed set, a max-sum-throughput vector always exists, but it may not be unique. Since $\mathcal{X}$ is convex, it contains a unique max-min fair vector \cite{radunovic}, and because 
the proportionally fair vector is the optimal solution  to maximizing the sum of logarithms (i.e. a strictly convex function) over $\mathcal{X}$, it exists and it is unique.

A connection between fairness and  convex optimization is rigorously established by the  problem of \emph{Network Utility Maximization} (NUM) \cite{Li05}: 
\begin{align}\label{eq:num}
\max_{\boldsymbol x \in \mathcal{X}} \sum_{k=1}^K g_{\alpha}(x_k),
\end{align}
where we have used the $\alpha$-fair function
\[
g_{\alpha}(x)=\left\{\begin{array}{ll}
\frac{x^{1-\alpha}}{1-\alpha},  &  \alpha\in [0,1)\cup (1,\infty)\\
\log x,  & \alpha = 1.
\end{array}\right.
\] 
Problem \eqref{eq:num} is useful since by tuning the value of $\alpha$ we obtain different fair vectors as solutions to the optimization problem: 
\begin{enumerate}
\item choosing $\alpha = 0$ yields max-sum-throughput, 
\item choosing  $\alpha= 1$ yields  proportional fairness \cite{kelly}, 
\item choosing $\alpha\to\infty$ yields max-min fairness \cite{mowalrand}.
\end{enumerate}

Consider the \emph{Gradient-Based Scheduling} (GBS) policy: schedule the user that maximizes $R_k(t)g_{\alpha}'(\overline{x}^{GBS}_k(t))$, where $g_{\alpha}'(x)=x^{-\alpha}$, and recall that $R_k(t)$ is the instantaneous transmission rate of user $k$ and $\overline{x}^{GBS}_k(t)$ is the accumulated throughput of user $k$; clearly GBS~$\in\Pi$.
Let $\boldsymbol x^*(\alpha)$ be a solution of \eqref{eq:num}, prior work \cite{Vijay,stolyar} has shown that $\overline{x}^{GBS}_k(t)\stackrel{\text{a.s.}}{\to}  x^*_k(\alpha),~\forall k$.
Hence, we can use GBS and tune the value of $\alpha$ to operate the system at any desirable point. There is anecdotal evidence that 3G and LTE base stations use GBS with $\alpha\approx 1$.



 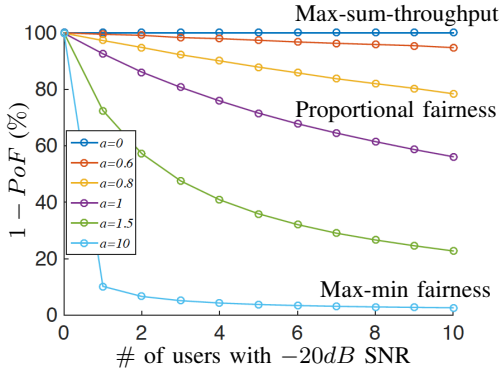
\begin{figure}[t!]
	 \centering
	 \begin{overpic}[width=0.33\textwidth]{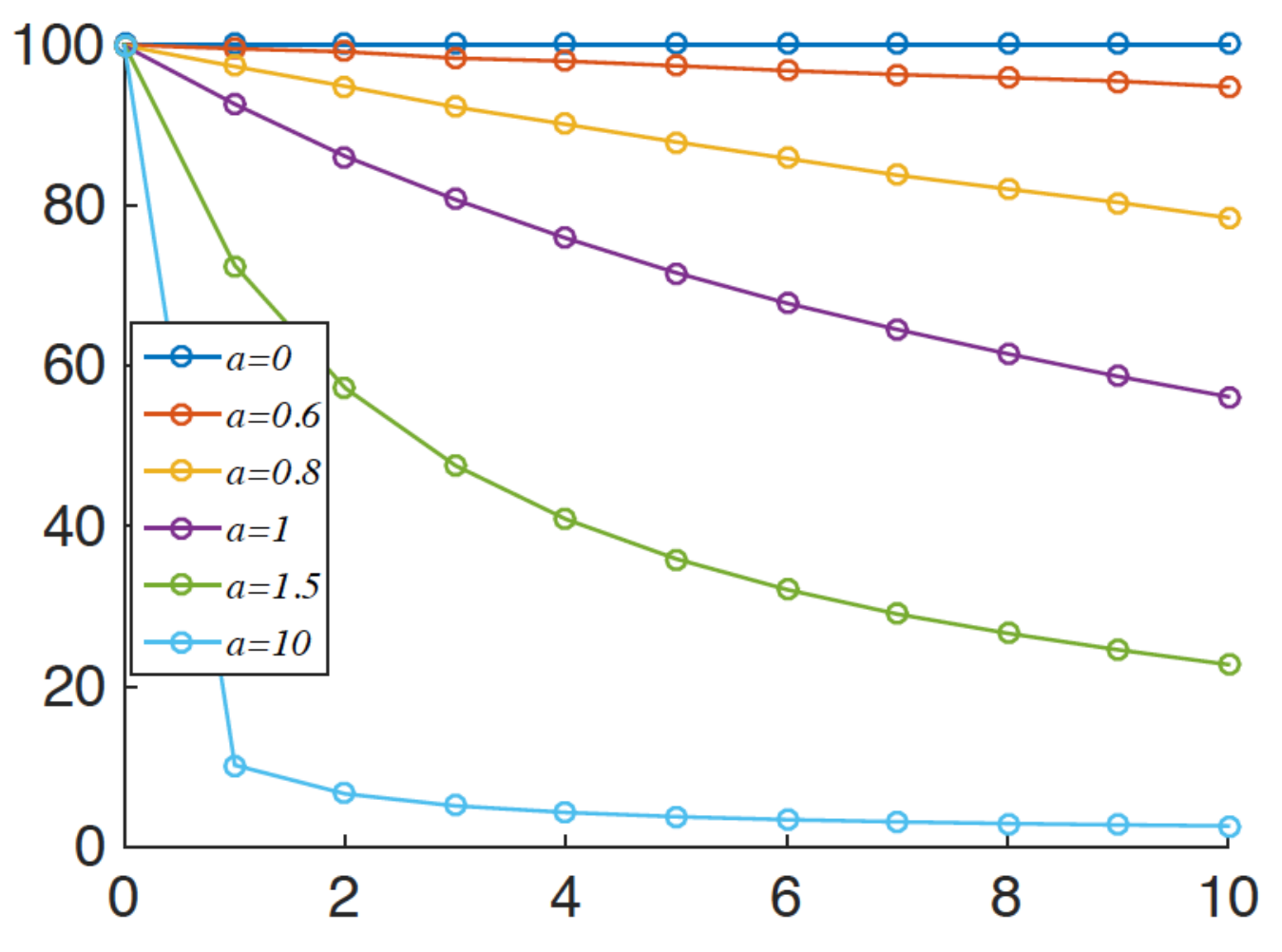}
	 	\put(21,-2.5){\small $\#$ of users with $-20dB$ SNR}
	 	\put(-3,25){\small\rotatebox{90}{$1-PoF$ ($\%$)}}
		\put(66.5,12){\small Max-min fairness}
		\put(61,52){\small Proportional fairness}
		\put(60.9,73.5){\small Max-sum-throughput}
	 \end{overpic}\vspace{0.1in} 
	 \caption{System throughput sensitivity ($1-PoF$) of alpha-fairness to users with poor channel quality. \vspace{-0.04in}
	}
	\label{fig:sensitivity}
\end{figure}

 \section{Optimizing Price of Fairness}\label{sec:sel}
 
Since $\boldsymbol x^*(\alpha)$ is the solution of \eqref{eq:num} for some $\alpha$, it follows that  $\boldsymbol x^*(0)$ denotes a max-sum-throughput vector. We define the price of fairness similar to \cite{pof}:
\begin{equation}
PoF\triangleq\frac{\sum_k x^*_k(0)-\sum_kx^*_k(\alpha)}{\sum_k x^*_k(0)},\label{eq:pof}
\end{equation}
which depicts how expensive it is to offer $\alpha$-fairness in terms of loss of sum throughput. Note that $1-PoF$  shows the fraction of maximum total throughput achieved by the  $\alpha$-fair point.

In wireless systems it is common for some users to have very poor signal reception. 
Such users can have a negative effect in $PoF$, which 
 we showcase with  a simulation example. In a downlink system with Rayleigh fading, we simulate GBS scheduling for $K=10$ users with mean channel rate $1$, while adding progressively 1..10 users with mean channel rate $20dB$ less.  We 
plot $1-PoF$  for different values of $\alpha$ in figure \ref{fig:sensitivity}. 
Proportional fairness experiences a significant  efficiency drop of almost 10$\%$ for 1 weak user, and 40$\%$ for 10. 

Our idea is to economize fairness by excluding weak users from  service. However, 
selecting users is challenging because we do not know \emph{a priori} their \emph{long-term} contributions to the throughput of a fairness-constrained system. 

\subsection{Selective fairness}

As a first step we introduce a novel fairness metric, called \emph{selective fairness}. 
We partition the set of users to two sets, $S\cup S^c=\{1,\dots,K\}$. For  set $S$ we guarantee $\alpha$-fairness, while  the users in the complementary set $S^c$ are not served at all.\footnote{The concept of selective fairness can be generalized to larger user partitions and different fairness objectives per subset.}
Our  goal will be  to  select the set $S$ carefully in order to decrease the price of fairness.

To concretely define selective fairness we need some technical tools. Define the   subspace $\mathcal{X}(S)\subseteq \mathcal{X}$, 
where all non-selected users  $S^c$ must have zero throughputs:
\begin{align*}
\mathcal{X}(S)\triangleq \left\{\boldsymbol x\in \mathcal{X}~\big|~ x_k=0, ~\forall k\in S^c\right\}.
\end{align*}
Also, we  need an $|S|$--dimensional representation of the  vectors in $\mathcal{X}(S)$.  
Let $\boldsymbol e_i$ be a $K$-dimensional column vector of zeros with the exception of element $i$ which is one, e.g. for $K=3$ we have $\boldsymbol e_2=(0,1,0)^T$. Note that $(\boldsymbol e_i)_{i\in \mathcal{K}}$ is a basis of $\mathbbm{R}^K$. Then define the $K\times |S|$ \emph{dimensionality reduction matrix} \vspace{-0.03in}
\[
\boldsymbol D_S\triangleq (\boldsymbol e_k)_{k\in S}.
\]
For the example of 3 users, and $S=\{1,3\}$, we have
\[
\boldsymbol D_{\{1,3\}}= \left(\begin{array}{cc}
1 & 0 \\
0 & 0 \\
0 & 1
\end{array}\right).
\]
If we multiply an element of $\mathcal{X}(S)$ or $\mathcal{X}$ with $\boldsymbol D_S$, we can remove the dimensions that correspond to $S^c$. Last, consider the set with the $|S|$--dimensional representations of $\mathcal{X}(S)$:
\begin{align*}
\Gamma(S)\triangleq \left\{\boldsymbol u~\big|~\boldsymbol u= \boldsymbol x \boldsymbol D_S , \boldsymbol x\in \mathcal{X}(S)\right\}.
\end{align*}

\begin{definition}[Selective fairness]
A vector $\boldsymbol x\in\mathcal{X}$ is called $(S,\alpha)$--\emph{selective fair} if
\begin{enumerate}
\item $\boldsymbol x\in\mathcal{X}(S)$,
\item and $\boldsymbol x \boldsymbol D_S $ is $\alpha$--\emph{fair} in the set $\Gamma(S)$. 
\end{enumerate}
\end{definition}

The definition posits that the selected users in $S$ will be allocated $\alpha$--\emph{fair} throughputs in the subspace $\Gamma(S)$, while the rest users receive zero throughput. 

Fix a subset $S\subseteq \{1,\dots, K\}$, and consider the conditions:
\begin{align}
& x_k=0, ~~\forall k\in S^c,\label{eq:sf1}\\
& \boldsymbol x\in\arg\max_{\boldsymbol u\in \mathcal{X}} \sum_{k\in S} g_{\alpha}(u_k).\label{eq:sf2}
\end{align}
\begin{theorem}
The vector $\boldsymbol x$ is $(S,\alpha)$--\emph{selective fair} if and only if it satisfies \eqref{eq:sf1}-\eqref{eq:sf2}.
\end{theorem}
\begin{proof}
Definition-1) is equivalent to  \eqref{eq:sf1}. We also establish the equivalence of definition-2) to \eqref{eq:sf2}:
\begin{align*}
& \boldsymbol x\in\arg\max_{\boldsymbol u\in \mathcal{X}} \sum_{k\in S} g_{\alpha}(u_k) \stackrel{\eqref{eq:sf1} \text{ or } \boldsymbol x\in\mathcal{X}(S)}{\Leftrightarrow}\\
& \boldsymbol x\in\arg\max_{\boldsymbol u\in \mathcal{X}(S)} \sum_{k\in S} g_{\alpha}(u_k) \Leftrightarrow \\
& \boldsymbol x \boldsymbol D_S \in\arg\max_{\boldsymbol u\in \Gamma(S)} \sum_{i=1}^{|S|} g_{\alpha}(u_i),
\end{align*}
{where the last follows by a  dimension permutation.}
%
\end{proof}

An implication of the theorem is that the $(S,\alpha)$--selective fair point is the limit point of GBS policy if we preclude users in $S^c$ from scheduling.


\subsection{$PoF$ minimization with   selective fairness}

In figure \ref{fig:sensitivity} we saw that when we block users with poor channel quality, the $PoF$ decreases.
It is therefore natural to ask the question \emph{if we were allowed to block all but $s_{\min}$ users in order to decrease $PoF$, which users would we block?} Next, we  pursue a $(S,\alpha)$-selective fair point that minimizes  $PoF$ (and thus maximizes system efficiency) subject to  serving at least  $s_{\min}$ users in a fair manner. 

\begin{align}
T(S^*)=&\max_{S\subseteq \{1,\dots,K\}} \sum_k x_k \label{eq:selfair1}\\
\text{s.t. }~& \boldsymbol x \in \arg\max_{\boldsymbol u\in \mathcal{X}} \sum_{k\in S} g_{\alpha}(u_k) \label{eq:c1}\\
& x_k=0, ~~k\in S^c\label{eq:c2}\\
& |S|\geq s_{\min},\label{eq:c3}
\end{align}
where, \eqref{eq:c1}-\eqref{eq:c2} ensure that the solution vector is $(S,\alpha)$-selective fair, 
the constraint \eqref{eq:c3} ensures that at least $s_{\min}$ users are served,
and the objective \eqref{eq:selfair1} aims  to achieve the maximum sum throughput denoted as $T(S^*)$ (from \eqref{eq:pof} this is equivalent to minimizing $PoF$). Problem  \eqref{eq:selfair1}-\eqref{eq:c3} is an admission control problem with fairness constraints.


We give a pictorial  example of $PoF$ minimization.
Consider a 3-user wireless downlink with feasible throughputs shown in figure \ref{fig:optselect}. The system operates with proportional fairness ($\alpha=1$) and must serve at least $s_{\min}=2$ users. Figure \ref{fig:optselect} shows all  possible selective fair points, where the points $(\{k\},1), k=1,2,3,\emptyset$ (indicated as dots with white interior) are infeasible due to constraint \eqref{eq:c3}. Among the feasible $(S,1)$-selective fair points, the optimization selects the point with the maximum total throughput, which in this case is $(\{1,2\},1)$. 
 \begin{figure}[t!]
\includegraphics[width=5.2cm]{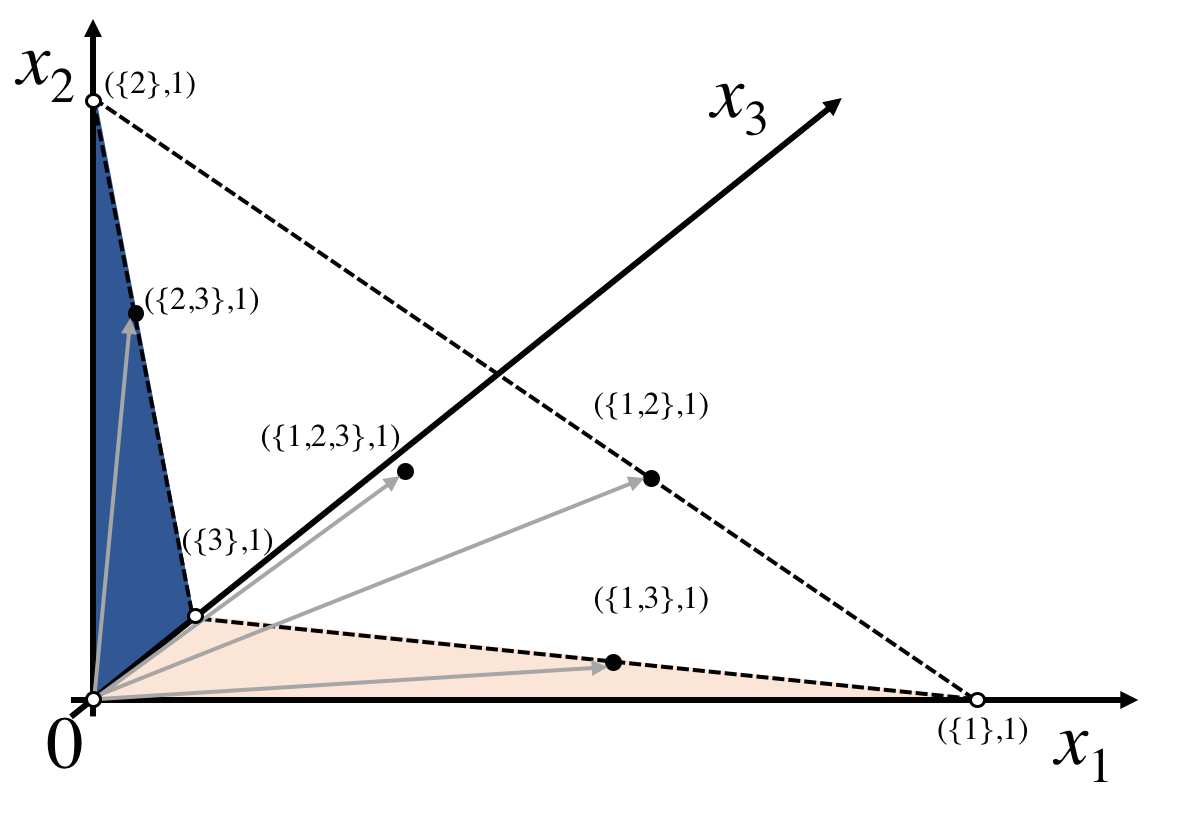}
\centering
\caption{3-user feasible throughputs and the eight possible $(S,1)$-selective fair points. 
}
\label{fig:optselect}
\end{figure}

\begin{corollary} 
Suppose a system is operated with a policy that solves \eqref{eq:selfair1}, and denote the total achieved throughput by $T_{\mathcal{K}}$ when serving users $\mathcal{K}$. If $|\mathcal{K}|\geq s_{\min}$, then:
\[
T_{\mathcal{K}}\leq T_{\mathcal{K}'},~~\forall \mathcal{K}\subseteq\mathcal{K}'.
\]
\end{corollary}
Hence, the optimization  \eqref{eq:selfair1} is useful because it ensures that the system performance does not drop even when adding  users with low average SNRs.


However $PoF$ minimization with selective fairness \eqref{eq:selfair1} is in general a combinatorial Mixed Integer Convex Program:  to determine the solution we  potentially need to inspect exponential to $K$ user subsets, and for each subset solve a difficult NUM problem. The NUM problem is difficult due to the possibly very large number of fading states, and explicit solutions for Rayleigh fading are only known for the case of max-min fairness \cite{Combes10}. 
Nevertheless, we present next a condition which is sufficient to break the combinatorial structure, and simplifies the solution of $PoF$ minimization. 

\subsection{Subspace monotonicity}
Consider a permutation of user indices $\sigma(.)$ and the arising $|S|$--dimensional subspace $\Gamma^{\sigma}(S)$, which is the same as $\Gamma(S)$ but with dimensions permuted by $\sigma$.
We can now compare the subspaces $\Gamma(S_1)$, $\Gamma(S_2)$ for two user sets with same cardinality $|S_1|=|S_2|=|S|$. 
We say that the subspace of $S_1$ dominates that of $S_2$ if \emph{there exists} permutation $\sigma_2$ such that 
\[
\Gamma(S_1)\supseteq \Gamma^{\sigma_2}(S_2).
\]

Figure \ref{fig:comparison} showcases a comparison of subspaces with $|S|=2$ on the example of figure \ref{fig:optselect}.

 \begin{figure}[t!]
\includegraphics[width=6cm]{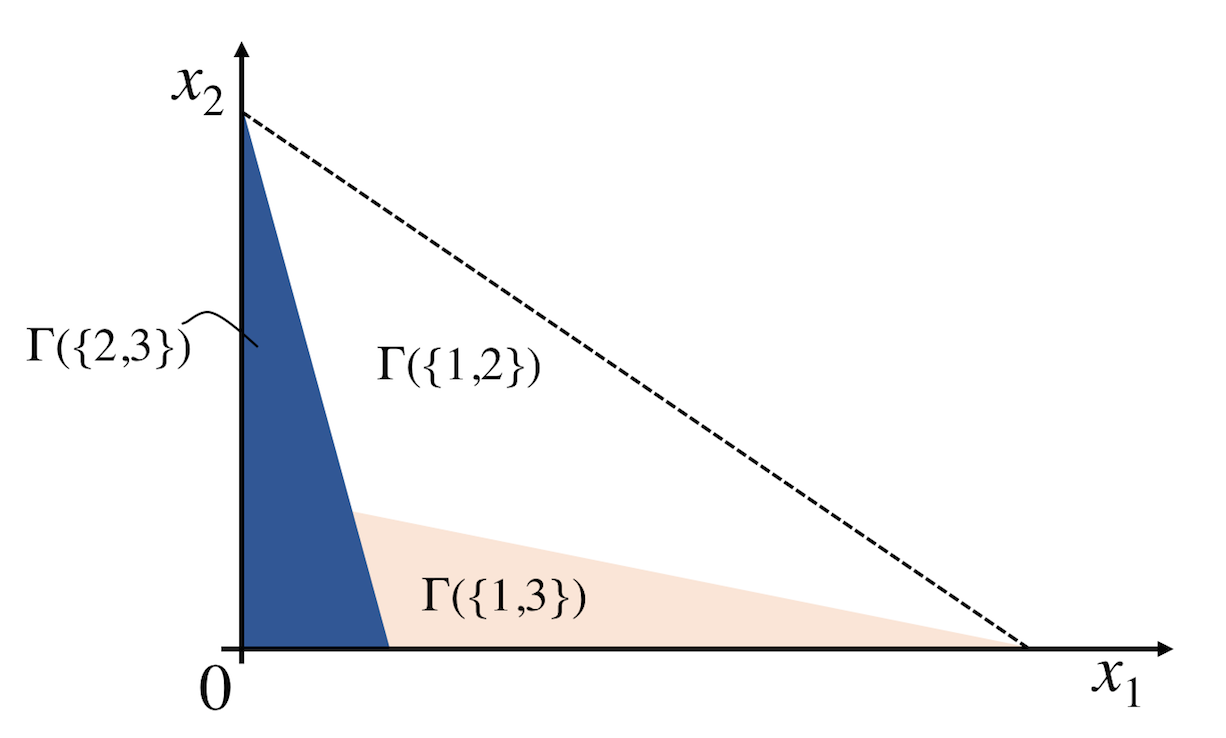}
\centering
\caption{Subspace comparison for the example of figure \ref{fig:optselect} and $|S|=2$.}
\label{fig:comparison}
\end{figure}

If a subspace dominates another, then it also contains a better selective fair vector;
let $T(S)$ be the total throughput of the $(S,\alpha)$--fair vector under a feasible $S$.
Suppose there exists permutation ${\sigma_2}$ such that  $\Gamma(S_1) \supseteq \Gamma^{\sigma_2}(S_2)$. Then $T(S_1)\geq T(S_2)$. Next we generalize this observation.

%
%


\begin{definition}[Subspace monotonicity]
For feasible throughputs $\mathcal{X}$ with dimension $K$, we say that the set  $\mathcal{X}$ satisfies the subspace monotonicity  if there exists an ordering of users $(\sigma(1),\sigma(2),\dots,\sigma(K))$ and permutations $\sigma_S$, such that: 
\[
\Gamma(\{\sigma(1),\dots,\sigma(|S|)\}) \supseteq \Gamma^{\sigma_S}(S),~~\forall S\subseteq \mathcal{K}.
\]
\end{definition}
Subspace monotonicity expresses the condition that if the  users are ordered according to $\sigma(.)$, the subspace of first $|S|$ users, denoted by $\Gamma(\{\sigma(1),\dots,\sigma(|S|)\})$, dominates the subspace of any other subset with same cardinality.

\begin{corollary}\label{cor:2}
Suppose that the set $\mathcal{X}$  satisfies the subspace monotonicity. Then an optimal solution of \eqref{eq:selfair1} is of the form $(x_{\sigma(1)}^*,x_{\sigma(2)}^*,\dots, x_{\sigma(j)}^*,0,\dots,0)$, where $s_{\min}\leq j \leq K$.
\end{corollary}

Therefore,  if the subspace monotonicity is satisfied,  then  problem \eqref{eq:selfair1} can be solved by a polynomial number of calls to an oracle that solves the NUM problem. Hence, we are motivated to ask: \emph{when is subspace monotonicity satisfied?}

For the interested user, subspace monotonicity is trivially satisfied for a resource allocation problem constrained on a simplex ($\sum_k x_k\leq 1$), where users can be ordered with marginal utilities. 
Next, we will show that it also applies to our problem under a specific condition for the channels. 
Recall that $R_k(t)$ is the user $k$ instantaneous channel rate, which is random  independently distributed across users and time. In the remaining of the paper we make the following assumption.
\begin{assumption}[Stochastic dominance]\label{as:stochDominance}
The channels are stochastically dominated: 
there exists a permutation $\sigma(.)$ of user indices, such that  $R_{\sigma(1)}(t)\geq_{\text{st}}\dots \geq_{\text{st}} R_{\sigma(K)}(t)$, where $R_{\sigma(i)}(t)\geq_{\text{st}}R_{\sigma(j)}(t)$ means 
\[
P\left(R_{\sigma(i)}(t)>x\right) \geq P\left(R_{\sigma(j)}(t)>x\right),~~ \forall x, t.
\]
\end{assumption}
Assumption \ref{as:stochDominance} is mild and holds for  many practical cases of interest, such as identically distributed fading channels with different means.

\begin{lemma}\label{lem:orderingLemma}
Consider a $K$-user wireless downlink  where Assumption \ref{as:stochDominance} holds. Then the feasible throughputs $\mathcal{X}$ satisfy the subspace monotonicity.
 \end{lemma} 
The proofs are in the Appendix.

\subsection{Minimizing $PoF$ with online experts}

We propose  an efficient online scheme to minimize $PoF$ \eqref{eq:selfair1}. Building on the subspace monotonicity property, our scheme is shown to achieve optimal performance using a small number of GBS experts. 
First we define the notion of \emph{Gradient-Based Scheduler (GBS) expert} for set $S\subseteq\mathcal{K}$. GBS($S$) uses the real system observations $\boldsymbol R(t)$ to simulate the $(S,\alpha)$-selective fair performance.

\noindent \rule[0.05in]{3.5in}{0.01in}

\vspace{-0.08in}
\noindent \textbf{GBS($S$) expert:}

\vspace{-0.06in}
\noindent \rule[0.05in]{3.5in}{0.01in}

\noindent\textbf{Initialize:} $\overline{x}^{S}_k(0)=0,~\forall k\in S$.  

\noindent\textbf{Iterate:} At $t$, observe $\boldsymbol R(t)$ and do:
\begin{enumerate}

\item \emph{Scheduling:} Choose user $k^*$ randomly from the set
\[
\arg\max_{k\in S}\frac{R_k(t)}{(\overline{x}^{S}_k(t))^{\alpha}}
\]

\item \emph{Throughput update:}
\[
\overline{x}^{S}_k(t+1)=
\frac{t}{t+1}\overline{x}^{S}_k(t) + \frac{1}{t+1}R_{k}(t)\mathbbm{1}_{(k=k^*)}
\]

\end{enumerate}

\vspace{-0.02in}
\noindent \rule[0.05in]{3.5in}{0.01in}

Observe that the GBS($S$) expert converges to $T(S)$ and hence ensures $\alpha$-fairness for the user set $S$. To minimize $PoF$, we next present the Selective GBS policy that precludes users from scheduling according to the best GBS($S$) expert.

\noindent \rule[0.05in]{3.5in}{0.01in}

\vspace{-0.08in}
\noindent \textbf{Selective GBS:}

\vspace{-0.06in}
\noindent \rule[0.05in]{3.5in}{0.01in}

\noindent\textbf{Input:} $s_{\min}$, permutation of user indices $\sigma(.)$.

\noindent\textbf{Initialize:} $\overline{x}^{{sel}}_k(0)=0,~\forall k=1,\dots,K$.  Consider the stronger user set per cardinality:
\begin{equation}\label{eq:setIndexing}
S_i = \{\sigma(1), \sigma(2),\dots,\sigma(i)\}, \forall i\in\{1,2,\dots,K\}
\end{equation}
Initialize all GBS($S$) experts  for $S=S_{s_{\min}},\dots,S_{K}$.

\noindent\textbf{Iterate:} At  $t$, observe $\boldsymbol R(t)$ and do:
\begin{enumerate}
\item \emph{Expert update:} Throughput update $\overline{x}^{S}_k(t)$ for each of the experts GBS($S_{s_{\min}}$), GBS($S_{s_{\min} + 1}$), $\dots$, GBS($S_{K} $).

\item \emph{Expert selection:} Choose the expert with highest total accumulated throughput 
\[
S^*(t)\in \arg\max_{\{S_i\}_{i\geq s_{\min}}} \sum_{k\in S_i}\overline{x}^{S_i}_k(t)
\]

\item \emph{Scheduling:} Choose user $k^*$ randomly from the set
\begin{equation}\label{eq:sel_sel}
\arg\max_{k\in S^*(t)}\frac{R_k(t)}{(\overline{x}^{sel}_k(t))^{\alpha}}
\end{equation}

\item \emph{Throughput update:}
\[
\overline{x}^{sel}_k(t+1) = \frac{t}{t+1}\overline{x}^{sel}_k(t) + \frac{1}{t+1}R_{k}(t)\mathbbm{1}_{(k=k^*)}
\]

\end{enumerate}

\vspace{-0.02in}
\noindent \rule[0.05in]{3.5in}{0.01in}



\begin{theorem}\label{th:selectiveGBSperformance}
Assume stochastically dominated channels and permutation $\sigma(.)$ such that corollary \ref{cor:2} holds. Then the Selective GBS converges almost surely to $T(S^*)$:
\[
\mathbb{P}\left\{ \lim_{t\rightarrow\infty}\sum_{k}\overline{x}^{sel}_k(t) = T(S^*)\right\} = 1.
\]
\end{theorem}

Notably from \eqref{eq:sel_sel}, the best expert is used to decide the user set restriction $S^*(t)$, but the scheduling decision of Selective GBS is ultimately made according to $\overline{x}^{sel}_k(t)$, which is different from what the experts see, $\overline{x}^{S}_k(t)$.

The Selective GBS is an online policy that does not require knowledge of fading statistics and can adapt to system changes. Additionally, it requires only $O(K)$ number of GBS($S$) experts, and hence a) selecting the best expert at each slot, and b) storing expert data in memory, is efficient. 

As input, Selective GBS requires the minimum number of accepted users $s_{\min}$ and the permutation $\sigma(.)$ that orders the users in decreasing average SNR. The latter can be kept updated using frequent measurements. Determining $s_{\min}$ is however complicated in practice; we should avoid blocking certain subscribers forever. The next section proposes our main result that deals with this issue: a novel stochastic framework which allows to \emph{determine the optimal set of blocked users on each scheduling realization while providing blocking frequency guarantees}.

\section{Stochastic selective fairness}

In wireless systems the actual number of users and their average channel quality varies from base station to base station and from time to time, hence it is impossible to operate selective fairness with a predetermined $s_{\min}$. 
Instead, we require a scheme that can adapt the number of selected users in an online fashion in order to a) improve system efficiency, and b) ensure Quality of Service. 


\subsection{Novel SLA}

We follow  a stochastic approach; we propose a novel Service Level Agreement (SLA), such that an arbitrary user is guaranteed to be selected with probability $1-\epsilon$, where $\epsilon$ is tunable. As opposed to the  snapshot scheduling of section \ref{sec:sel}, we consider a sequence of  scheduling problem instances $\boldsymbol z_1,\boldsymbol z_2,\dots$, where each $\boldsymbol z_n$ is a realization of the random variable $\boldsymbol{Z}$ describing the spatial distribution of the users, that is the set of active users $\mathcal{K}(\boldsymbol{Z})$ and their average SNRs due to slow fading, denoted by $\overline{SNR}_k$ for user $k\in\mathcal{K}(\boldsymbol{Z})$.
Then let $y_k(\boldsymbol{Z}) = 1$ indicate the event that user $k$ is selected at a random realization $\boldsymbol{Z}$, we require the probabilistic constraint:
\[
\mathbb{P}_{\boldsymbol{Z}}\left(y_k(\boldsymbol{Z}) = 1 | k\in\boldsymbol{Z} \right)\geq 1 - \epsilon, \forall k.
\]
The users are considered statistically equivalent, hence the constraint intuitively  captures the behavior of an ergodic subscriber, i.e., one who buys a service with the associated SLA, and then uses the service over a long time horizon each time from possibly different locations. In this context, the constraint implies that such a subscriber will be rarely blocked.
\vspace{0.2in}
 \begin{figure}[h!]
\includegraphics[width=0.5\textwidth]{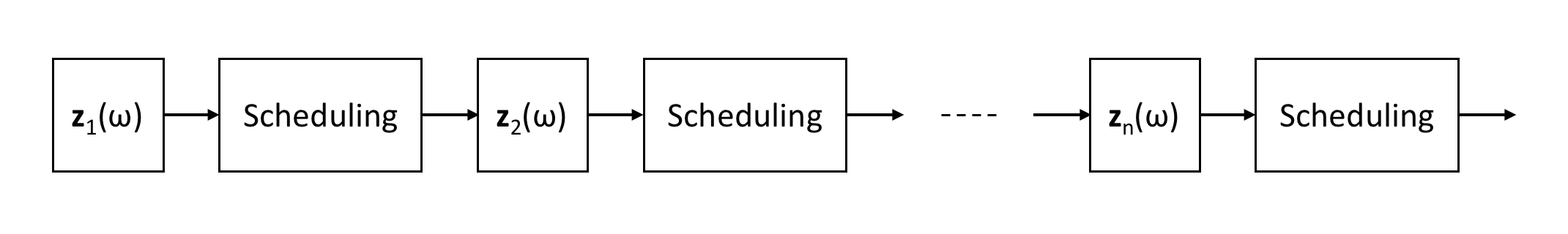}
\centering
\caption{Scheduling over multiple realizations.\vspace{0in}}
\label{fig:realizations}
\end{figure}

\subsection{Problem formulation}

In this section we formalize a $PoF$ minimization over multiple scheduling realizations. Let $\mathcal{T}(n) = \{t_{n},...,t_{n+1}-1\}$ be the set of time slots spanned by the $n-$th realization of the spatial distribution, $\boldsymbol{Z}_n$. Then, given the realization of the spatial distribution in $\mathcal{T}(n)$, $\boldsymbol{z}_n$, we solve a selective-fair scheduling problem where we need to decide which users to admit in the system for these slots; let this decision be denoted by $\boldsymbol{y}(\boldsymbol{z}_n)$. 
We need the following assumption for time scale separation: 

\begin{assumption}[Time scale separation]\label{as:timeScaleSeparation}
	The changes in spatial realizations occur in a slower time scale than the fast fading. In particular, the duration of a realization $\mathcal{T}(n)$ is long enough for the empirical throughput averages to converge to the $\alpha$-fair rates, that is 
	\[
	\left|\frac{1}{|\mathcal{T}(n)|}\sum_{t=t_n}^{t_{n+1}-1}\mu_k(t) - T_k^*(\boldsymbol{z}_n,\boldsymbol{y}(\boldsymbol{z}_n))\right| \leq \delta 
	,\]
	for some small $\delta>0$, where $T_k^*(\boldsymbol{z}_n,\boldsymbol{y}(\boldsymbol{z}_n))$ is the optimal selective $\alpha$-fair rate vector for spatial realization $\boldsymbol{z}_n$, given the admission decision $\boldsymbol{y}(\boldsymbol{z}_n)$.
\end{assumption} 
The above assumption implies that user arrivals/departures and slow fading fluctuations occur slower than the convergence of the scheduler, which in practice takes a few seconds. 
 

Let the admission policy $\boldsymbol{y}(.)$ be a function from spatial realizations to deciding if a user will be admitted or not. 
Our general problem then is to find the admission policy  to solve the following stochastic optimization:
\begin{align}\label{eq:stoprob}
\max_{\boldsymbol{y}}\quad & \mathbb{E}_{\boldsymbol{Z}}\left\{\sum_{k}T_k^*(\boldsymbol{Z}, \boldsymbol{y}(\boldsymbol{Z}))\right\}\\ \label{eq:stoconst}
\text{s.t.}\quad & \mathbb{P}_{\boldsymbol{Z}}\left(y_k(\boldsymbol{Z}) = 1 | k\in\boldsymbol{Z} \right)\geq 1 - \epsilon, ~\forall k,  
\end{align} 
where the constraint \eqref{eq:stoconst} ensures the satisfaction of the SLA of each subscriber, and the maximal value of \eqref{eq:stoprob} represents the best way the system can use the available admission budget to minimize the $PoF$. Note that this minimization is very complicated since it  jointly considers 
i) the actual number of active users at each realization, 
ii) the random user locations and corresponding average SNRs, and
iii) 
the cost of fairness. 

We proceed to reformulate the problem as an optimization of time-averages. First, it will be useful to transform the per-user SLA constraints into an equivalent constraint for the total number of admitted users:

\begin{lemma}\label{lem:equivalenceLemma}
	Assume the spatial distribution of all users is the same. Then, \eqref{eq:stoconst} is equivalent to 
	\begin{equation*} \label{eq:equivalence} 
	\mathbb{E}_{\boldsymbol{Z}}\left\{\sum_{k\in\mathcal{K}(\boldsymbol{Z})}y_{k}(\boldsymbol{Z})\right\} \geq (1-\epsilon)\mathbb{E}_{\boldsymbol{Z}}\left\{|\mathcal{K}(\boldsymbol{Z})|\right\}
	\end{equation*}
\end{lemma}

Then, we use the \emph{law of large numbers} to express the objective function as a time average over scheduling realizations:
\begin{align*}
&\mathbb{E}_{\boldsymbol{Z}}\left\{\sum_{k}T_k^{*}(\boldsymbol{Z}, \boldsymbol{y}(\boldsymbol{Z}))\right\} =\\
 &\hspace{0.4in}\lim_{N\rightarrow\infty}\frac{1}{N} \mathbb{E}\left\{\sum_{n=0}^{N-1}\sum_{k\in \boldsymbol{z}_n}T_k^*(\boldsymbol{z}_n, \boldsymbol{y}(\boldsymbol{z}_n))\right\},
\end{align*}
where the expectation is taken with respect to possible random control (if the control is deterministic it can be eliminated).
Finally, constraint \eqref{eq:stoconst} can be rewritten as a time average of indicator functions:
\begin{equation}\label{eq:constraintTimeAvg}
\lim_{N\rightarrow\infty}\frac{1}{N}\mathbb{E}\left\{\sum_{n=0}^{N-1}\mathbbm{1}_{(k\in \boldsymbol{Z}_n)}y^n_k(\boldsymbol{Z}_n)\right\} \geq (1-\epsilon) \mathbb{E}_{\boldsymbol{Z}}\left\{|\mathcal{K}(\boldsymbol{Z})|\right\}
.\end{equation}


The problem to solve now has become:
\begin{align}
\max_{\{\boldsymbol y^n\}_n,\eqref{eq:constraintTimeAvg}} \lim_{N\rightarrow\infty}\frac{1}{N} \mathbb{E}\left\{\sum_{n=0}^{N-1}\sum_{k\in \boldsymbol{z}_n}T_k^*(\boldsymbol{z}_n, \boldsymbol{y}^n(\boldsymbol{z}_n))\right\}.  \label{eq:problem}
\end{align}


\section{Stochastic Selective Fair Policies}

The form of \eqref{eq:problem}  motivates a stochastic optimization approach \cite{georgiadis06}, where we can use a virtual queue to ensure the satisfaction of the time-average constraint. 

Define a \emph{virtual queue} $Q(n)$, which evolves at the same  time scale as the spatial process as follows: 
\begin{equation}\label{eq:vQueue}
Q(n+1) = \left[Q(n) + A(n) - D(n)\right]^+
,\end{equation}
where $A(n)$ is defined as 
\begin{equation} \label{eq:virtualArrivals}
A(n) = \begin{cases}
		|\mathcal{K}(\boldsymbol{Z}_n)|, \text{w.p.} \quad 1-\epsilon \\
		0, \text{w.p.} \quad \epsilon,
\end{cases}
\end{equation}

and $
D(n) = D(\boldsymbol{Z}_n, \boldsymbol{y}^n(\boldsymbol{Z}_n)) = \sum_{k\in\mathcal{K}(\boldsymbol{Z}_n)}y^n_k(\boldsymbol{Z}_n)
$
is the number of users admitted at the $n-$th spatial realization. Note that $A(n)$ is a  random variable  with mean $(1-\epsilon)\mathbb{E}_{\boldsymbol{Z}}\left\{|\mathcal{K}(\boldsymbol{Z})|\right\}$. The idea then is that, if the queue is stable, the mean of its \emph{service} $D(n)$ will be larger than the one of its its \emph{arrivals} $A(n)$, and the equivalent SLA constraint \eqref{eq:constraintTimeAvg} will be met.
Therefore, our policy will strive to stabilize $Q(t)$.
 We observe that the virtual queue can be seen as a counter to track/learn the Lagrange multiplier corresponding to the SLA constraint.

\subsection{Known selective fair throughputs}
We  first deal with the case where given a realization $\boldsymbol z_n$, the  corresponding selective-fair sum throughputs $\sum_kT_k^*( \boldsymbol{z}_n, \boldsymbol{y})$ are known for each $\boldsymbol y$. 

\noindent \rule[0.05in]{3.5in}{0.01in}

\vspace{-0.08in}
\noindent \textbf{Drift Plus Penalty (DPP):}

\vspace{-0.06in}
\noindent \rule[0.05in]{3.5in}{0.01in}

\noindent \textbf{Initialize:} Fix parameter $V>0$. 

\noindent \textbf{Iterate over scheduling realizations:} At the beginning of $t_n$: 
\begin{enumerate}
	\item \emph{User selection:}
\begin{equation}\label{eq:idppMaximization}
		\boldsymbol{y}^n = \arg\max_{\boldsymbol{y}}\left[ \sum_kT_k^*( \boldsymbol{Z}_n, \boldsymbol{y}) + \frac{Q(n)}{V}\sum_{k\in\mathcal{K}(\boldsymbol{Z}_n)}y_k \right].
	\end{equation}
	\item \emph{Virtual queue update:} Set $D(n) = \sum_ky_k^n$, draw a random variable $A(n)$ and update the queue as in \eqref{eq:vQueue}.  
\end{enumerate}

\vspace{-0.06in}
\noindent \rule[0.05in]{3.5in}{0.01in}


\begin{theorem}\label{th:optimalityIDPP}
	The DPP satisfies the SLAs, and yields sum throughput within $O(1/V)$ of the maximum in \eqref{eq:problem}. 
\end{theorem}

\subsection{Unknown selective fair throughputs}

In practice determining $\sum_kT_k^*( \boldsymbol{z}_n, \boldsymbol{y})$ for a given set of users $\boldsymbol{y}$ knowing their channel statistics is very challenging, with the exception  of proportional and max-min fairness, where approximations exist \cite{Combes10}, and the case of max throughput where the solution is to always admit all users. Therefore,  we cannot always apply the  DPP  policy. 
In this section we propose a combination of Selective GBS (which requires $s_{\min}$) and DPP (which requires $\sum_kT_k^*( \boldsymbol{z}_n, \boldsymbol{y})$), to design an policy which progressively learns  the best user set and thus  does not need the information of neither $s_{\min}$, nor $\sum_kT_k^*( \boldsymbol{z}_n, \boldsymbol{y})$. 
The only information needed is the ordering of the users according to their SNRs, which is found by measurements. 


\noindent \rule[0.05in]{3.5in}{0.01in}

\vspace{-0.08in}
\noindent\textbf{Online Selective Fair (OSF) Scheduler:}

\vspace{-0.06in}
\noindent \rule[0.05in]{3.5in}{0.01in}

\noindent\textbf{Initialize:} Fix parameter $V>0$, virtual queue service $S^*(0)=S$. 

\noindent\textbf{Iterate $t_n$ (over scheduling realizations):} 
\begin{enumerate}
	\item \emph{Update virtual queue:} Generate the random variable $A(n)$, as per \eqref{eq:virtualArrivals}, set $D(n) = |S^*(t_n)|$, i.e. the cardinality of the subset chosen in the last slot of last realization, and update the queue: 	$Q(n+1) = \left[Q(n) - D(n) + A(n)\right]^+$.
	\item \emph{User ordering:} Consider current user realization $\boldsymbol z_n$, and $K_n=|\mathcal{K}(\boldsymbol z_n)|$. Then permute users with $\sigma(.)$ such that $\overline{SNR}_{\sigma(1)}\geq \overline{SNR}_{\sigma(2)}\geq...\geq \overline{SNR}_{\sigma(K_n)}$.
	
	\item \emph{Initialize next realization:} Reset $t=0$, construct sets ${S}_i$ as per \eqref{eq:setIndexing}, 
	set $\overline{x}_k(0) = 0, \forall k$ and $\overline{x}^{{S}_i}_k(0)=0, \forall k\in{S}_i$ for each expert $i$, 	and 	run the next realization using the following policy:
	 \end{enumerate} 
	
\noindent \textbf{Iterate $t$:}
\begin{enumerate} 
	\item \emph{Expert update:} Execute one step of each expert $\text{GBS}({S}_i), i\in\{1,2,...,K_n\}$ using $\boldsymbol R(t)$. Let $\overline{\boldsymbol{x}}^{{S}_i}(t+1)$ be the vector of accumulated throughputs of expert $i$.
	
	\item \emph{Expert selection:} Choose the expert that maximizes:
	\[
	{S}^*(t) = \arg\max_{\{{S}_i\}_i}\left[\sum_{k\in{S}_i}\overline{{x}}^{{S}_i}_k(t+1) +|{S}_i|\frac{Q(n)}{V}\right]
	\]
	\item \emph{Scheduling:} Choose user $k^*(t)$ randomly from the set:
	\[
	\arg\max_{k\in{S}^*(t)}\frac{R_{k}(t)}{(\overline{x}_k(t))^{\alpha}}
	\]
	\item \emph{Throughput update:}
	\[
		\overline{x}_k(t+1) = \frac{t}{t+1}\overline{x}_k(t) + \frac{1}{t+1}R_{k}(t)\mathbbm{1}_{(k = k^*(t))}
	\]
\end{enumerate}

\vspace{-0.02in}
\noindent \rule[0.05in]{3.5in}{0.01in}


\begin{corollary}[Optimality of OSF]\label{th:onlineOptimality}
	Let assumptions of i) time-scale separation, ii) statistically identical users, and iii)  stochastically dominated channels hold.
			Theorem \ref{th:selectiveGBSperformance} implies that at each scheduling realization the experts converge to the corresponding selective fair throughputs, hence OSF chooses the action that maximizes \eqref{eq:idppMaximization}, and by Theorem \ref{th:optimalityIDPP} 
	\emph{OSF satisfies the SLA constraints and 
	yields sum throughput within $O(1/V)$ of the maximum in \eqref{eq:problem}}. 
\end{corollary}

A few remarks about OSF operation: searching for the best expert and storing data for all experts is $O(K)$, hence OSF has low-complexity. 
Also, to determine $D(n)$, we need an estimate of the number of users admitted in last realization. For this, we have used the cardinality of the selected users by the best expert at the last slot. Theoretically, due to the time-scale separation, all the experts and the scheduler have converged and hence this is indeed the number of users served in that realization. In our simulations, the convergence is not completely reached but this method remains accurate.


\section{Numerical Results}

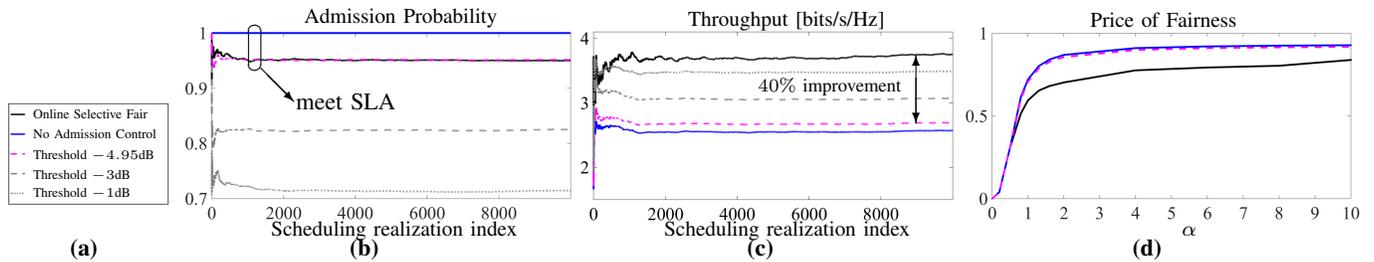
\begin{figure*}[t!]
	\centering
 \hspace{-0.65cm}	\subcaptionbox{}[.15\linewidth][c]{
 	\vspace{0.2cm}
		\begin{overpic}[scale=0.135]{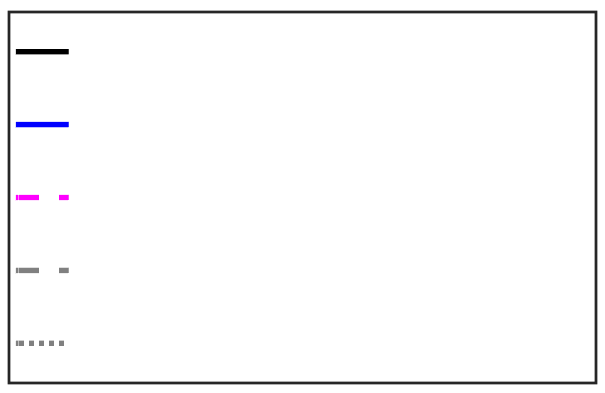}
			\put(16,54){\tiny{Online Selective Fair}}
			\put(16,42){\tiny{No Admission Control}}
			\put(16,29){\tiny{Threshold $-4.95$dB}}
			\put(16,17){\tiny{Threshold $-3$dB}}
			\put(16,6){\tiny{Threshold $-1$dB}}
		\end{overpic}
	}\hspace{-0.08in}
	\subcaptionbox{}[.27\linewidth][c]{%
			\begin{overpic}[scale=0.17]{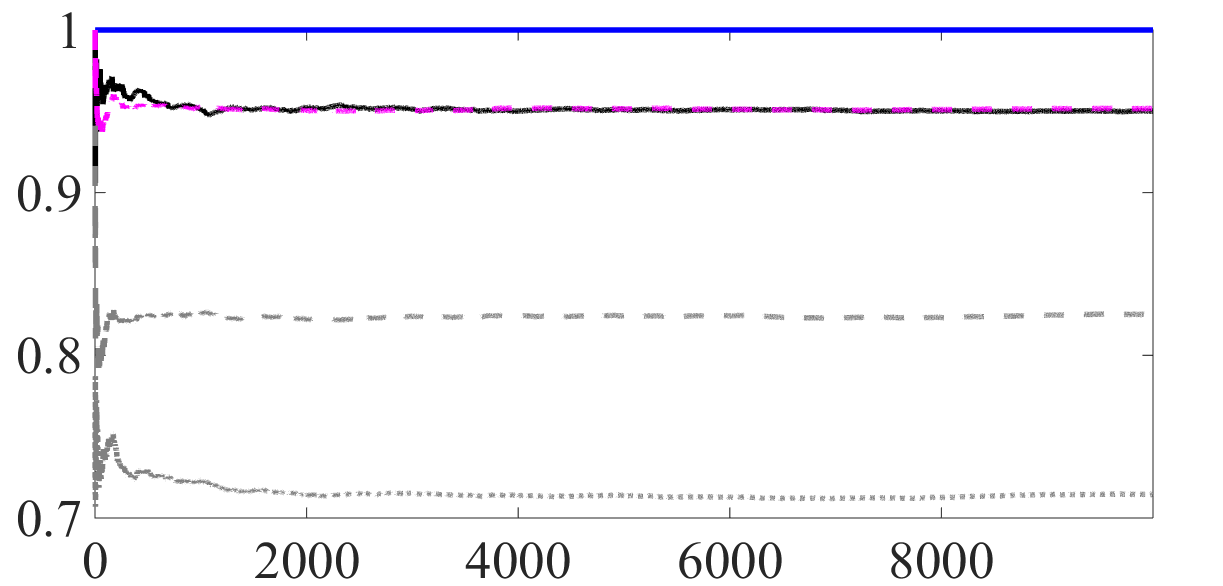}
			\put(22,-3.5){\footnotesize Scheduling realization index}
			\put(30,48){\footnotesize Admission Probability}
			\put(18, 42){\oval(3, 11)}
			\put(19.5,37){\vector(1,-1){8}}	
			\put(28,27){\small meet SLA}
		\end{overpic}
	}~~
	\subcaptionbox{}[.27\linewidth][c]{%
		\begin{overpic}[scale=0.17]{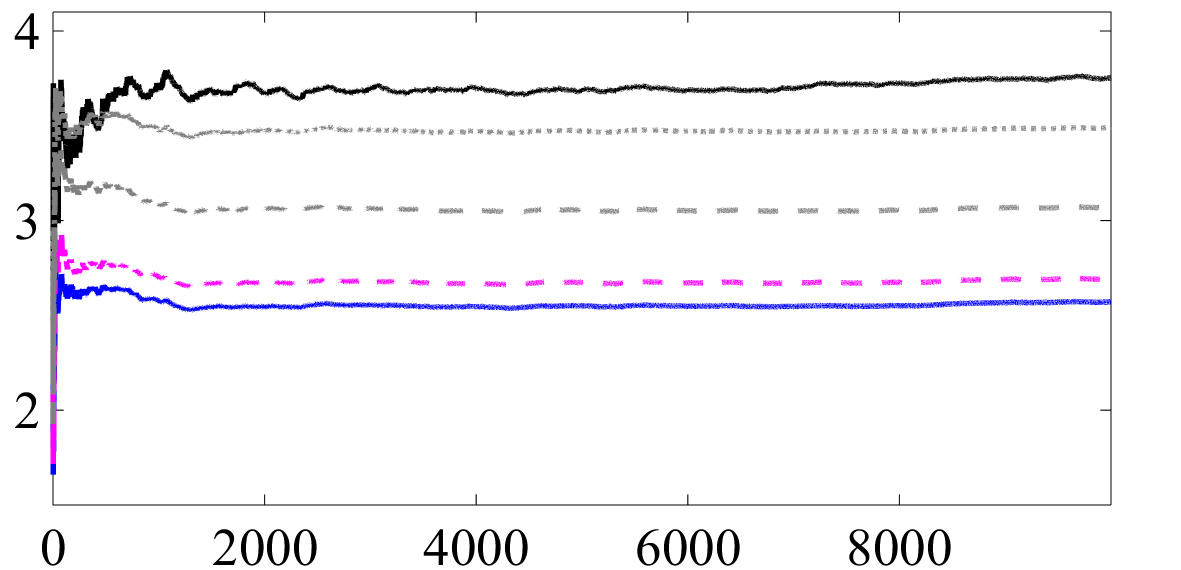}
			\put(22, -3.5){\footnotesize Scheduling realization index}
			\put(28,49){\footnotesize Throughput [bits/s/Hz]}		
			\put(85,25){\vector(0,1){16.7}}	
			\put(85,36){\vector(0,-1){11.5}}
					
			\put(45.5, 32.5){\scriptsize $40\%$ improvement}
		\end{overpic}
	}~
	\subcaptionbox{}[.27\linewidth][c]{%
		\begin{overpic}[scale=0.17]{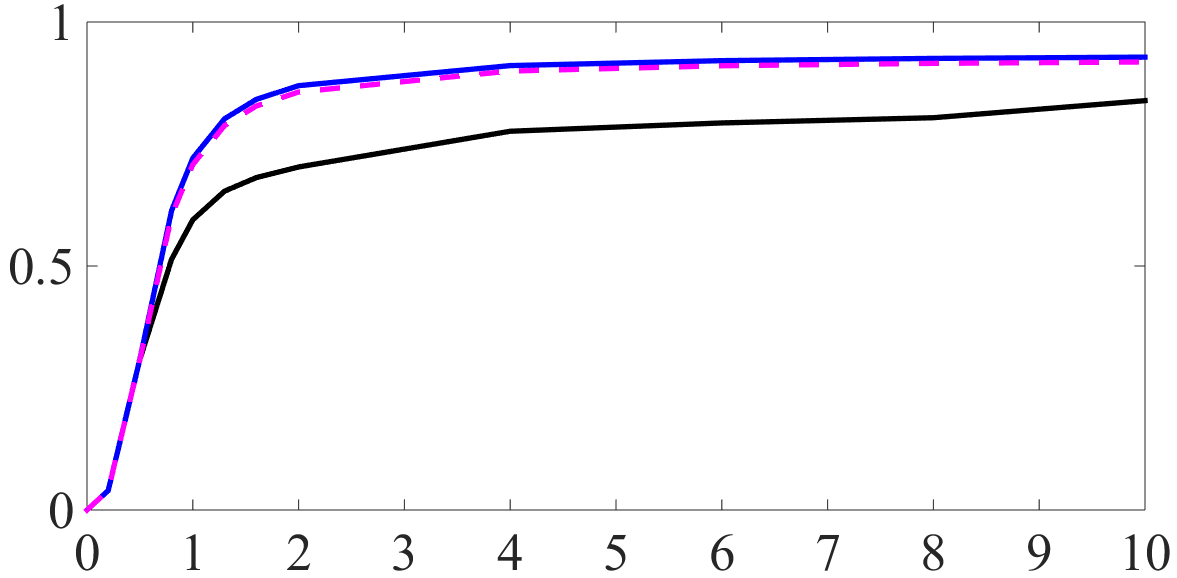}
			\put(55,-3.5){\footnotesize $\alpha$}
			\put(32,49){\footnotesize {Price of Fairness}}		
		\end{overpic}
	}
	\vspace{0.1in}
	\caption{Simulation results for blocking probability $5\%$: (a) Plot legends. (b) Empirical admission probabilities for $\alpha=1$. (c) Empirical average system throughput for $\alpha=1$. (d) Price of Fairness for different levels of fairness (values of $\alpha$).}\label{fig:SimResults095}
\end{figure*}

In this section we evaluate the performance of OSF against policies that block users with average SNR below a threshold, and against the case where no admission control is done and all users are served in a fair manner. We consider a system with $100$ subscribers, each requesting a probabilistic SLA of being served $95\%$ of the time. The system evolves as per figure \ref{fig:realizations}. At each scheduling realization, a subscriber is active with probability $0.1$ and is placed uniformly at random in a cell. The parameters are chosen such that a user at the edge of the cell has an average SNR of $-5dB$. Each scheduling realization lasts for $3000$ time slots, and at each slot the channel realizations are drawn from a Rayleigh fading distribution.      

The results are presented in figure \ref{fig:SimResults095}.  Figures \ref{fig:SimResults095}b-\ref{fig:SimResults095}c show the case $\alpha=1$; the OSF policy achieves the SLA (\ref{fig:SimResults095}b) and outperforms the best naive threshold admission control that also achieves the SLA by approximately $40\%$ in terms of total throughput (\ref{fig:SimResults095}c). In addition, it outperforms the policy that does no admission control by a wider margin. Figure \ref{fig:SimResults095}d compares $PoF$  for different values of $\alpha$. In low values of  $\alpha$, where fairness constraints are loose, all policies have similar performance. As $\alpha$ increases, having a good user admission strategy makes fairness cheaper; the $PoF$  of OSF  increases  slower than the other two policies.  \emph{We attribute the gains of our approach to accurately considering the actual impact of blocking to $PoF$.}

We also experimented with more stringent values of the SLA. For $\alpha=1$  the gain in system throughput with respect to the best threshold  policy that satisfies the SLA is  $10\%$ and $2\%$ for corresponding selection probabilities $99\%$ and $99.9\%$. Decreasing gain for more stringent SLA is to be expected, since for selection probability 1 all the policies will yield the same $PoF$-they are all forced to accept all users at all times.  However, we mention that  the probability of good coverage  of LTE is measured in drive tests to be 90-95$\%$ \cite{EricssonLTE}, which induces a fairly large amount of blocking. 


\section{Conclusions}

In this paper, we introduced \emph{selective fairness}, the idea of providing fairness to some users and blocking the others, in order to mitigate the adverse effect of users with very poor channel quality. Extending this concept to the stochastic setting, we derived an online policy that maximizes the system's throughput subject to satisfying an SLA on the user blocking probability. Our intelligent blocking outperforms by $40\%$ naive approaches that simply block low-SNR users.
It is worth noting that our results can be easily extended to multiple SLA classes using one queue for each class, as well as systems with multiple cells using an operator-wide global queue. Analyzing the transient behavior of the selective GBS with experts and extending our results to cases where users have different spatial distributions are very interesting topics for further study. 

\bibliographystyle{IEEEtran}
\bibliography{bibliography}

\newpage
\appendix
\begin{IEEEproof}[Proof of Lemma \ref{lem:orderingLemma}]
	Without loss of generality we will compare two sets of users $\mathcal{K} = \{1,\dots,K-1,K\}$ and $\mathcal{K}' = \{1,\dots,K-1,K'\}$ (here indices are not indicative of channel quality), where the  the channel state probability distribution of all users is the same except of user $K'$, $q_l$,  which dominates that of $K$, $p_l$:
	\[
	\sum_{l=1}^m q_l \geq \sum_{l=1}^m p_l,~~m=1,\dots,L.
	\]
	and for weights  $\beta_1\geq \beta_2 \geq \dots \geq \beta_L\geq 0$, it follows
	\begin{equation}\label{eq:posdom}
	\sum_{l=1}^L q_l\beta_l \geq \sum_{l=1}^L p_l\beta_l,  
	\end{equation}
	which we will use later.

	Let $\Gamma(\mathcal{A})$ be the feasible rate region of user set $\mathcal{A}$. Then we should show that $\Gamma(\mathcal{K}) \subseteq \Gamma(\mathcal{K}')$. If we prove this pairwise comparison, then the subspace monotonicity follows by extending to all sets of same cardinality using the stochastic dominance order.

	Define $\Gamma(\mathcal{K}, \boldsymbol{s})$ as the set of feasible throughput vectors given that the channel states of  users $\{1,\dots,K-1\}$ are fixed to $\boldsymbol{s}\in\mathcal{R}^{K-1}$ and only the channels of the last user vary.
	More specifically, we will show that $\Gamma(\mathcal{K}, \boldsymbol{s}) \subseteq \Gamma(\mathcal{K}', \boldsymbol{s}), \forall \boldsymbol{s}$, from which the result follows since 
	\[
	\Gamma(\mathcal{K}) = \sum_{\boldsymbol{s}}\pi(\boldsymbol{s})\Gamma( \mathcal{K},\boldsymbol{s}) \subseteq  \sum_{\boldsymbol{s}}\pi(\boldsymbol{s})\Gamma( \mathcal{K}',\boldsymbol{s}) = \Gamma(\mathcal{K}').
	\]
	where $\pi(\boldsymbol{s}) = \Pi_{k=1}^{K-1}\mathbb{P}(R_k(t) = r_k(\boldsymbol{s}))$ is the probability of state $\boldsymbol{s}$.

	Fix some $\boldsymbol{s}\in\mathcal{R}^{K-1}$, and choose a vector $\boldsymbol u\in \Gamma(\mathcal{K}, \boldsymbol{s})$, we want to show that $\boldsymbol u\in \Gamma(\mathcal{K}', \boldsymbol{s})$. 
	Denote $\phi_{k,l}$ the fraction of time user $k$ is scheduled when the channel state of the last user is $l$, where  $\phi_{k,l}\in[0,1]$ and $\sum_k \phi_{k,l} \leq 1, \forall l$. The following hold: 
	\begin{align*}
	\left\{\begin{array}{l}
	u_k= \sum_{l}p_l\phi_{k,l}r_k(\boldsymbol{s}), ~~ \forall k < K \\ 
	u_K= \sum_{l}p_l\phi_{K,l}R^{l}.  
	\end{array}\right.
	\end{align*}
	Here, let us observe that since the channel states $\boldsymbol{s}$ are fixed, the throughputs $u_k, k<K$ depend only on the total fraction of time assigned to these users $1-\sum_l p_l \phi_{K,l}$ irrespective of state $l$. To this end, consider the optimization problem:
	\begin{align}
	\min_{[\phi_{K,l}]_{l=1,\dots,L}} & \sum_l p_l \phi_{K,l}\label{eq:opt}\\
	\text{s.t. } & u_K= \sum_{l}p_l\phi_{K,l}R^{l}\notag\\
	& \phi_{K,l}\in[0,1],\notag
	\end{align}
	where the solutions of this problem ensure both the achievability of throughput $u_K$ (due to first constraint) but also that there remain enough time fractions to achieve $u_k, k<K$ (due to the objective). Let $\boldsymbol \phi^*$ be a solution to \eqref{eq:opt}, we have
	\[
	\sum_l p_l \phi_{K,l}^* \leq  \sum_l p_l \phi_{K,l},
	\]
	hence the time users $k<K$ are scheduled $1-\sum_l p_l \phi_{K,l}^*$ is large enough   such that there must exist coefficients $(\phi_{k,l}^*)_{k<K}$ for which:
	\begin{equation}
	u_k= \sum_{l}q_l\phi_{k,l}^*r_k(\boldsymbol{s}), ~~ \forall k < K. \label{eq:uk}
	\end{equation}
	
	Additionally, due to the form of the optimization problem \eqref{eq:opt}, we must also have that  $\phi_{K,1}^*\geq \phi_{K,2}^* \geq \dots \geq \phi_{K,L}^*$, hence it is also
	$\phi_{K,1}^*R^1\geq \phi_{K,2}^*R^2 \geq \dots \geq \phi_{K,L}^*R^L$, and using \eqref{eq:posdom} it follows
	\begin{align*}
	\sum_lq_l\phi_{k,l}^*R^l\leq   \sum_lp_l\phi_{k,l}^*R^l,
	\end{align*}
	and hence
	\[
	u_K= \sum_{l}q_l\phi_{K,l}^*R^{l},
	\]
	which combined with \eqref{eq:uk}  proves that $\boldsymbol u\in \Gamma(\mathcal{K}', \boldsymbol{s})$, and concludes the proof that $\Gamma(\mathcal{K}, \boldsymbol{s}) \subseteq \Gamma(\mathcal{K}', \boldsymbol{s}), \forall \boldsymbol{s}$.
\end{IEEEproof}

\begin{IEEEproof}[Proof of Theorem \ref{th:selectiveGBSperformance}]
	We assume that $T(S_i)\neq T(S_j)$, define $\mathcal{B}_i(\delta) = \left(T(S_i) - \delta, T(S_i) + \delta\right)$ and choose $\delta_0>0$ such that 
	\[
		\mathcal{B}_i(\delta_0) \cap \mathcal{B}_j(\delta_0) = \emptyset, \forall i,j \in\{s_{min},\dots, K\} 
	.\]
	We also define the following events:  
	\[
	\mathcal{A}_i^{\delta}(t) = \left\{\left|\sum_{k\in S_i}\overline{x}^{S_i}_k(\tau)  - T(S_i)\right|\leq \delta, \forall \tau \geq t\right\}
	.\] By the almost sure convergence of GBS \cite{Vijay,stolyar}, we have: for every $\epsilon>0$, there exists a finite $M_0(\epsilon)$ such that  
	\begin{equation}\label{eq:M0}
	\mathbb{P}\left\{\mathcal{A}_i^{\delta_0}(t)\right\} > \sqrt[K]{1-\epsilon}, \forall t>M_0(\epsilon), \forall i 
	.\end{equation}
	 This, together with how we chose $\delta_0$, implies that $\forall t\geq M_0(\epsilon)$:
	 	 \begin{equation}\label{eq:selectiveGBSperfProof}
	 \mathbb{P}\left\{S^*(\tau) = S^*, \forall \tau\geq t \right\} \geq \prod_{i=s_{min}}^K\mathbb{P}\left\{\mathcal{A}_i^{\delta_0}(t)\right\}   > 1-\epsilon.
	 \end{equation}
	 Now, define the event that determines the convergence:
	 \[
	 \mathcal{C}^{\delta}(t) = \left\{\left|\sum_{k}\overline{x}^{sel}_k(\tau) - T(S^*)\right| \leq \delta , \forall \tau\geq t\right\},
	 \]
	 and fix any $\epsilon,\delta > 0$. From the almost sure convergence property of GBS, we know that if for some $S^*(t) = S^*, \forall t\geq \tau_0$ for some $\tau_0$ then there exists a constant $M_1(\delta)$ such that the above event is true for all $t_0 \geq \tau_0 + M_1(\delta)$. In addition, from \eqref{eq:selectiveGBSperfProof} we have that $S^*(t) = S^*, \forall t\geq \tau_0$ with probability  at least $1-\epsilon$ for $\tau_0 \geq M_0(\epsilon)$, as defined in $\eqref{eq:M0}$. The above discussion implies that, there exists a $t'_0(\epsilon, \delta) = M_0(\epsilon) + M_1(\delta)$ (any value greater than this quantity will also do) such that 
	 \[
	 \left|\mathbb{P}\left\{\mathcal{C}^{\delta}(t)\right\} - 1\right| < \epsilon, \forall t\geq t'_0(\epsilon,\delta)
	 ,\]which implies $\lim_{t\rightarrow\infty}\mathbb{P}\left\{\mathcal{C}^{\delta}(t)\right\} = 1, \forall \delta > 0$.
\end{IEEEproof}

\begin{IEEEproof}[Proof of Lemma \ref{lem:equivalenceLemma}]
 Due to the identical spatial distribution for all users, it follows that for any $\boldsymbol{y}\in\Pi$ and any two users $k,k'\in\mathcal{K}$ we must have:
	 \[
	 \mathbb{P}_{\boldsymbol{Z}}\left(y_k(\boldsymbol{Z}) = 1 | k\in\boldsymbol{Z} \right) = \mathbb{P}_{\boldsymbol{Z}}\left(y_k'(\boldsymbol{Z}) = 1 | k'\in\boldsymbol{Z} \right).
	 \] 
	 Let $p(\boldsymbol{y})= \mathbb{P}_{\boldsymbol{Z}}\left(y_k(\boldsymbol{Z}) = 1 | k\in\boldsymbol{Z} \right) $, conditioning on $\left(k\in\boldsymbol{Z}\right)$  we obtain: 
	 \begin{align} \nonumber
	 &\mathbb{E}_{\boldsymbol{Z}}\left\{\sum_{k\in\mathcal{K}(\boldsymbol{Z})}y_{k}(\boldsymbol{Z})\right\} =\\
	 & \hspace{0.4in}= \sum_{k}\mathbb{P}_{\boldsymbol{Z}}\left(k\in\boldsymbol{Z}\right)\mathbb{P}_{\boldsymbol{Z}}\left(y_k(\boldsymbol{Z}) = 1 | k\in\boldsymbol{Z} \right) \nonumber\\ 
	 & \hspace{0.4in}= p(\boldsymbol{y})\sum_{k}\mathbb{P}_{\boldsymbol{Z}}\left(k\in\boldsymbol{Z}\right) = p(\boldsymbol{y})\mathbb{E}_{\boldsymbol{Z}}\left\{|\mathcal{K}(\boldsymbol{Z})|\right\}\notag
	 .\end{align}
	 Eq. \eqref{eq:stoconst} then implies that $p(\boldsymbol{y}) \geq 1-\epsilon$, yielding the desired.
\end{IEEEproof}

\begin{IEEEproof}[Proof of Theorem \ref{th:optimalityIDPP}]
	Let $\boldsymbol{y}^*(\boldsymbol z_n),~n=1,2,\dots$ denote the optimal solution of \eqref{eq:stoprob}-\eqref{eq:stoconst}, i.e. a (possibly randomized) function from a scheduling realization to a user admission decision; existence of optimal stationary randomized policies is established in \cite{neely10}. 
	To prove optimality of DPP, we will compare to the performance of $\boldsymbol{y}^*(\boldsymbol z_n)$ with respect to the Lyapunov drift. Indeed, we first define the drift of policy $\boldsymbol{y}$:
	\[
		\Delta(\boldsymbol{y}, \boldsymbol{Z}, Q) =  \sum_{k\in\mathcal{K}(\boldsymbol{Z})}\left(VT^*_k(\boldsymbol{Z}, \boldsymbol{y}) + Qy_k\right).  
	\]
	Recall that $\boldsymbol{y}^n$ denotes the decision of DPP at realization $n$, and observe that it is designed to minimize the quantity $\Delta(.)$.
	Then, for a given scheduling realization $\boldsymbol{z}_n$ and given value of the virtual queue $Q(n)$, we can show that 
	\begin{align} \nonumber
	\mathbb{E}\big\{&Q^2(n+1) - Q^2(n) \big\} - V\mathbb{E}\left\{\sum_k T^*_k(\boldsymbol{z_n}, \boldsymbol{y}^n(\boldsymbol{z}_n))\right\}  \\ \nonumber 
	&\leq B + Q(n)\mathbb{E}\left\{A(\boldsymbol{z}_n)\right\} - \mathbb{E}\left\{\Delta(\boldsymbol{y}^n, \boldsymbol{z}_n, Q(n))\right\} \\ \nonumber 
	&\leq B + Q(n)\mathbb{E}\left\{A(\boldsymbol{z}_n)\right\} - \mathbb{E}\left\{\Delta(\boldsymbol{y}^*, \boldsymbol{z}_n, Q(n))\right\}
	,\end{align} 
	where $B$ is a constant that depends only the parameters of the system. Additionally, from optimality of $\boldsymbol{y}^*$ we have 
	\begin{align} \nonumber
	\mathbb{E}\big\{Q^2(n+1) - Q^2(n) \big\} - V\mathbb{E}\left\{\sum_k T^*_k(\boldsymbol{z_n}, \boldsymbol{y}^n(\boldsymbol{z}_n))\right\}  \\ \nonumber 
	\leq B - V \mathbb{E}\left\{\sum_{k\in\mathcal{K}\boldsymbol{Z})}T^*_k(\boldsymbol{z_n}, \boldsymbol{y}^*(\boldsymbol{z}_n))\right\} 
	,\end{align}
	from which we can deduce that (i) $Q(n)$ is (mean rate) stable therefore the SLA constraints hold and (ii) the value of the average throughput less than the optimal by at most $B/V$, finishing the proof. 
\end{IEEEproof}

\end{document}